\documentclass[journal]{IEEEtran}
\usepackage{hyperref}
\usepackage{graphicx}
\usepackage{subfigure}
\usepackage{enumerate}
\usepackage{amsmath}
\usepackage{color}
\usepackage{amsmath}
\usepackage{algorithm}
\usepackage[noend]{algpseudocode}
\usepackage{multirow}
\usepackage{amsmath,amssymb,amsthm,mathrsfs,amsfonts,dsfont}

\newcommand{\bp}{\begin{proof} \small }
\newcommand{\ep}{\end{proof} \normalsize}
\newcommand{\epx}{\end{proof} \small}
\newcommand{\bpa}{\begin{proofappx} \footnotesize }
\newcommand{\epa}{\end{proofappx} \small }

\newtheorem{proposition}{Proposition}
\newtheorem{corollary}{Corollary}
\newtheorem{lemma}{Lemma}
\newtheorem{assumption}{Assumption}

\newtheorem*{theorem*}{Theorem}
\newtheorem*{proposition*}{Proposition}
\newtheorem*{corollary*}{Corollary}
\newtheorem*{lemma*}{Lemma}
\newtheorem*{assumption*}{Assumption}
\newtheorem*{definition*}{Definition}
\newtheorem*{claim*}{Claim}

\newcommand{\be}{\begin{equation}}
\newcommand{\ee}{\end{equation}}
\newcommand{\bs}{\begin{subequations}}
\newcommand{\es}{\end{subequations}}
\newcommand{\bq}{\begin{eqnarray}}
\newcommand{\eq}{\end{eqnarray}}
\newcommand{\bqn}{\begin{eqnarray*}}
\newcommand{\eqn}{\end{eqnarray*}}

\newcommand{\ba}{\left[ \begin{array}}
\newcommand{\ea}{\\ \end{array} \right]}
\newcommand{\ben}{\begin{enumerate}}
\newcommand{\een}{\end{enumerate}}

\def\real{{\mathchoice%
{\hbox{\rm\setbox1=\hbox{I}\copy1\kern-.45\wd1 R}}
{\hbox{\rm\setbox1=\hbox{I}\copy1\kern-.45\wd1 R}}
{\hbox{\scriptsize\rm\setbox1=\hbox{I}\copy1\kern-.45\wd1 R}}
{\hbox{\scriptsize\rm\setbox1=\hbox{I}\copy1\kern-.45\wd1 R}}}}

\def\Zint{{\mathchoice{\setbox1=\hbox{\sf Z}\copy1\kern-.75\wd1\box1}
{\setbox1=\hbox{\sf Z}\copy1\kern-.75\wd1\box1}
{\setbox1=\hbox{\scriptsize\sf Z}\copy1\kern-.75\wd1\box1}
{\setbox1=\hbox{\scriptsize\sf Z}\copy1\kern-.75\wd1\box1}}}
\newcommand{\complex}{ \hbox{\rm C\kern-0.45em\rule[.07em]{.02em}{.58em}%
\kern 0.43em}}

\begin{document}
%
\title{Personalized Course Sequence Recommendations}
%
%
%

\author{Jie~Xu,~\IEEEmembership{Member,~IEEE,}
        Tianwei~Xing,~\IEEEmembership{Student Member,~IEEE,}
        and~Mihaela~van~der~Schaar,~\IEEEmembership{Fellow,~IEEE}
\thanks{J. Xu is with the Department
of Electrical and Computer Engineering, University of Miami, Coral Gables,
FL 33146, USA. Email: jiexu@miami.edu. }
\thanks{T. Xing and M. van der Schaar are with the Department of Electrical Engineering, University of California, Los Angeles, CA 90095, USA. Emails: twxing@ucla.edu, mihaela@ee.ucla.edu. }}

\maketitle

\begin{abstract}
Given the variability in student learning it is becoming increasingly important to tailor courses as well as course sequences to student needs. This paper presents a systematic methodology for offering personalized course sequence recommendations to students. First, a forward-search backward-induction algorithm is developed that can optimally select course sequences to decrease the time required for a student to graduate. The algorithm accounts for prerequisite requirements (typically present in higher level education) and course availability. Second, using the tools of multi-armed bandits, an algorithm is developed that can optimally recommend a course sequence that both reduces the time to graduate while also increasing the overall GPA of the student. The algorithm dynamically learns how students with different contextual backgrounds perform for given course sequences and then recommends an optimal course sequence for new students. Using real-world student data from the UCLA Mechanical and Aerospace Engineering department, we illustrate how the proposed algorithms outperform other methods that do not include student contextual information when making course sequence recommendations.
\end{abstract}


%
\IEEEpeerreviewmaketitle

\section{Introduction}
Recent studies \cite{cca2014fouryear}\cite{astin2005degree} find that the vast majority of college students in the United States do not complete college in four years and that fewer college students are today graduating on time than a decade ago. Taking longer to graduate is not cheap - it costs \$15,933 more in tuition, fees and living expenses for every extra year at a public two-year college and \$22,826 for every added year at a public four-year college \cite{cca2014fouryear}. While many factors contribute to students taking longer to graduate, such as credits lost in transfer, uninformed choices due to the low advisor-student ratios and poor preparation for college, the inability of students to take required courses when needed is among the leading causes \cite{cca2014fouryear}. If courses are elected and taken myopically, without a clear plan, a student may end up in an awkward situation in which required subsequent courses are offered (much) later, thereby (significantly) prolonging graduation time. To reduce the time-to-graduation, it is therefore of paramount importance for the student to elect courses in a foresighted way by taking into account the possible subsequent course sequences (including which courses are mandatory and which ones are not, and the course prerequisites) and when the various courses are offered. More importantly, because the number and variety (in backgrounds, in knowledge, in goals) of students is expanding rapidly, it is more and more important to tailor course sequences to students since the same learning path is unlikely to best serve all students. Therefore, it is necessary to develop an automated course sequence recommendation system that learns from the performance of previous students in various courses/sequences and uses what it has learned to adaptively recommend  course sequences that are personalized for the current student, depending on the student's background and his/her completion status of the program in order to maximize any of a variety of objectives including time to graduation, grades and the trade-off between the two. Issuing personalized course sequences recommendations for students poses numerous challenges, some of which are unique to this type of recommendation system:
\begin{itemize}
\item \textbf{Sequences} Unlike most existing recommendation systems (such as those used to recommend movies or products to purchase), course sequence recommendations requires issuing sequences of courses (items) rather than a single item at a time. Hence, such course sequence recommendations require dealing with a large decision space which grows combinatorially with the number of courses. Searching for the best sequence to make personalized recommendations to a student represents a challenging problem even when the number of courses is moderate. When the number of offered courses is large (as it is typical in academic programs at the undergraduate or graduate level), offering personalized recommendations becomes a very challenge problem.
\item \textbf{Flexibility and Constraints} There is a great deal of flexibility in course sequence recommendation since multiple courses can be taken simultaneously. At the same time, course sequence recommendation is also subject to many constraints - some courses are mandatory while some are not, and some courses are prerequisite of others. Both the flexibility and the constraints make course sequence recommendation an extremely complex problem.
\item \textbf{Evolving Recommendation} Any static course sequence is sub-optimal since the knowledge, experience and performance of a student develops and evolves in the process of learning. Students may arrive at different completion status of the academic program at a given time depending on their performance in the finished courses and hence, they should be given different recommendations of subsequent courses.
\item \textbf{Personalization} An even more difficult challenge is that students vary tremendously in backgrounds, knowledge and goals. As a result, the same recommendation policy is unlikely to best serve all students and provide the best learning experience for every student. However, personalizing course sequence recommendation seems a daunting task since students attend college only once and hence, it is unrealistic to try different course sequences on a student and learn the best course sequence for this student.
\end{itemize}
\begin{figure}
  \centering
  \includegraphics[scale = 0.6]{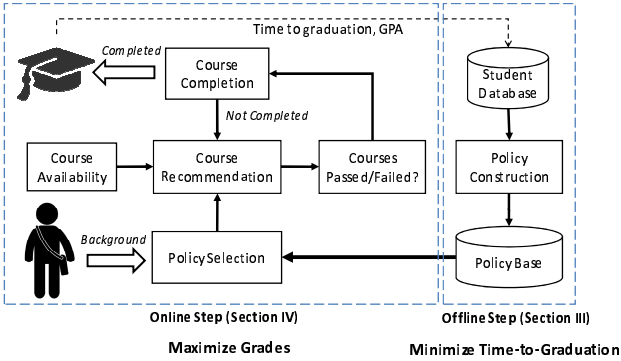}\\
  \caption{System Diagram for Course Sequence Recommendation}\label{System}
\end{figure}

In this paper, we develop an automated course sequence recommendation system that is able to provide personalized and adaptive recommendation to students depending on their background as well as their evolving performance in the program. In order to reduce complexity and enable tractable solutions, we solve this problem in two steps. The first step (Section III) involves offline learning, in which a set of candidate recommendation policies are determined to minimize the expected time to graduation or maximize the on-time graduation probability using an existing dataset of anonymized student records. A dynamic programming based approach is adopted to solve the adaptive sequence recommendation which recommends subsequent course sequences to students depending on their completion status of the academic program by taking into account the prerequisite relationship among courses and the course availability across academic terms (quarters/semesters). The second step (Section IV) involves online learning, in which for each new student, a suitable course sequence recommendation policy is selected depending on this student's background using the learned knowledge from the previous students. Online contextual multi-armed bandit techniques are used to develop policy selection algorithms to maximize the students' grades given the time-to-completion constraints. To enable fast learning, the algorithm exploits the similarity among students and adaptively clusters students and refines the clustering as more students enter and finish the program. See Figure \ref{System} for a depiction of the system.

The main contributions of this paper are as follows:
\begin{itemize}
\item We formulate the personalized and adaptive course sequence recommendation problem and develop systematic solutions aimed at reducing time to graduation and maximizing students' grades upon graduation.
\item We provide analytical characterizations on the impact of course prerequisite dependency and the course availability on the possible emerging subsequent course sequences as well the structure of the optimal course sequence recommendation policy in certain specific scenarios.
\item We rigorously analyze the performance of the online personalized policy selection algorithm and prove that the proposed algorithm converges fast and is able to select the optimal personalized course sequence recommendation policy for students.
\item Extensive simulations are carried out on a real-world student record dataset to verify the efficacy of the proposed system. They also reveal how the time-to-graduation is affected by the course prerequisite dependency as well as the course availability across academic terms, thereby providing important insights and guidelines on how to plan the curriculum and allocate teaching resources to improve the on-time graduation ratio.
\end{itemize}

The rest of this paper is organized as follows. In Section II, we review the literature and highlight our contribution. In Section III, we study the offline learning problem to determine a set of candidate course sequence recommendation policies based on dynamic programming. In Section IV, we study the online learning problem to select personalized policies for students based on online contextual multi-armed bandits techniques. Section V provides simulation results. In Section VI, we conclude the paper.

\section{Related Work}
Machine learning for education has recently gained much attention \cite{baraniuk2015open}\cite{Cen2015Big}. Previous research focuses on grade prediction \cite{meier2015predicting}, drop-out prediction \cite{KDDCUP}, personalized teaching styles and materials \cite{tekin2015etutor}, estimating learners' knowledge of concepts underlying a domain \cite{lan2014sparse}, multimedia and cooperative learning \cite{asif2004multimedia} etc. This paper studies the important, yet much less investigated problem of (personalized) course sequence recommendation. Solving this problem has the potentially significant impact of shortening the time that students need to graduate. Methods solving this problem can then be combined with other methods to provide a comprehensive set of tools for personalizing education.

There is much work on recommending relevant courses/learning materials to students according to students' types (e.g. interests, knowledge levels, learning styles and feedback) \cite{klavsnja2011learning} \cite{chen2005personalized} \cite{wen2008effective} \cite{farzan2006social} \cite{bendakir2006using} \cite{ghauth2010learning} \cite{shishehchi2011review}. Besides course recommendation, there is extensive work on recommender systems for assisting users with finding desirable products or services \cite{resnick1997recommender} \cite{balabanovic1997fab} \cite{ricci2011introduction}. However, several unique features of course sequence recommendation make these approaches unsuitable for the considered problem. First, while traditional recommendation systems deal with the problem of recommending items or sets of items, most of them do not take into account prerequisites while recommending an item: a course can be taken only when all its prerequisite courses have been taken and passed. Thus, it does not make sense to recommend to a student a course if the prerequisite courses have not been completed. Secondly, there are complex constraints on the recommendation: the courses taken by a student must satisfy requirements (e.g. take 10 mandatory courses and 5 out of 12 elective courses) in order for the student to graduate. Thirdly, courses are not available in all quarters due to teaching resource constraints. If the next available quarter of a required course is far away from the current quarter, then it might be wiser to take this course earlier rather than later.

Recommendation with prerequisites was studied in \cite{parameswaran2010evaluating}, in which the goal is to recommend the best set of $k$ items when there is an inherent ordering between items. Various prerequisite structures were studied and the complexity of determining the best set is proven to be NP-Hard. Several heuristic approximation algorithms were developed to solve the recommendation problem. However, the problem is formulated as a set recommendation problem rather than a sequential recommendation problem, which ignores the course prerequites, the evolving knowledge of students (and grades so far) as well as the course availability in different quarters.  Recommendation with complex constraints was studied in \cite{parameswaran2011recommendation} where increasingly expressive models were developed to check if the requirements are satisfied and course recommendations were made by taking into account these requirements. However, the course prerequisites and the course availability are not considered. A Markov Decision Process based recommender system was developed in \cite{shani2002mdp} to take into the long-term effects of each recommendation. However, this approach is not able to handle the course prerequisites or the course requirement constraints.

Our algorithm for online personalized recommendation policy selection builds on the contextual multi-armed bandits methods \cite{slivkins2014contextual} \cite{dudik2011efficient} \cite{langford2008epoch} \cite{chu2011contextual}. Most of the prior work on contextual bandits is focused on an agent making single-stage decisions based on the provided context information. In contrast, in this paper, arms are the course sequence recommendation policies which are selected depending on the student's background but the policy itself also induces a sequence of decision making depending on the evolving performance of the student.

\section{Course Sequence Recommendation: Policy Construction}

We consider a curriculum consisting of a set of courses $\mathcal{N} = \{1,2,...,N\}$. Among these courses, there are $M$ mandatory courses and $E = N-M$ elective courses. We consider a discrete time system where a student takes courses quarter by quarter \footnote{We use the quarter system for illustration but our approach also works for the semester system.} and can stay in the program for at most $T$ quarters. Quarters are indexed by $t =  1,2,...,T$. Let $s(t)$ be a course state vector of size $N$ which is used to indicate the courses that a student has already taken and passed by the end of quarter $t$. The first $M$ elements are with respect to the mandatory courses and the remaining $E$ elements are with respect to the elective courses. Each element of $s(t)$ takes a binary value where $s_n(t) = 1$ means that the student has taken and passed course $n$ and $s_n(t) = 0$ otherwise. Initially each student passes zero courses such that each element of $s(0)$ satisfies $s_n(0) = 0, \forall n$.

In each quarter, the maximum number of courses a student can take is $C$. Let $A(t)$ denote the elected courses in quarter $t$. However, even though the student is free to elect courses, $A(t)$ must come from a set of feasible courses $\mathcal{F}(t,s(t-1))$ depending on the index $t$ of the quarter as well as the course state $s(t-1)$ of the student. Firstly, courses that have already been taken and passed cannot be retaken. Secondly, the student can only take courses that are available in that quarter since a course is usually not offered in every quarter. Let $\Gamma(t)\subseteq \mathcal{N}$ denote the set of available courses offered in quarter $t$. Thirdly, the student can only take a course when he/she has finished all its prerequisite courses. We formalize course prerequisite in more detail as follows.

\textbf{Course Prerequisite} Courses have prerequisite dependencies, namely courses can be elected only when certain prerequisite courses have been taken and passed. In general, the prerequisite dependency can be described as a directed acyclic graph (DAG), denoted by $\mathcal{G} = \langle \mathcal{N}, \mathcal{E}\rangle$ where $\mathcal{N}$ is the set of courses and $\mathcal{E}$ is the set of directed edges. A directed edge $m\to n$ between two courses $m$ and $n$ means that course $m$ is a prerequisite of course $n$. Let $P(n) = \{m: m\to n \in \mathcal{E}\}$ be the set of prerequisite courses of $n$. Only when all courses in $P(n)$ have been taken can course $n$ be elected. Note that if $P(n)$ is an empty set, then course $n$ has no prerequisite courses and hence can be elected at any time (whenever available). Moreover, any elective course cannot be the prerequisite course of a mandatory course. Otherwise, the elective course effectively becomes mandatory. Nevertheless, an elective course can be the prerequisite of another elective course. Figure \ref{Prerequisite} illustrates part of the prerequisite graph for the undergraduate program in the Mechanical and Aerospace Engineering department at UCLA.

\begin{figure}
  \centering
  \includegraphics[scale = 0.45]{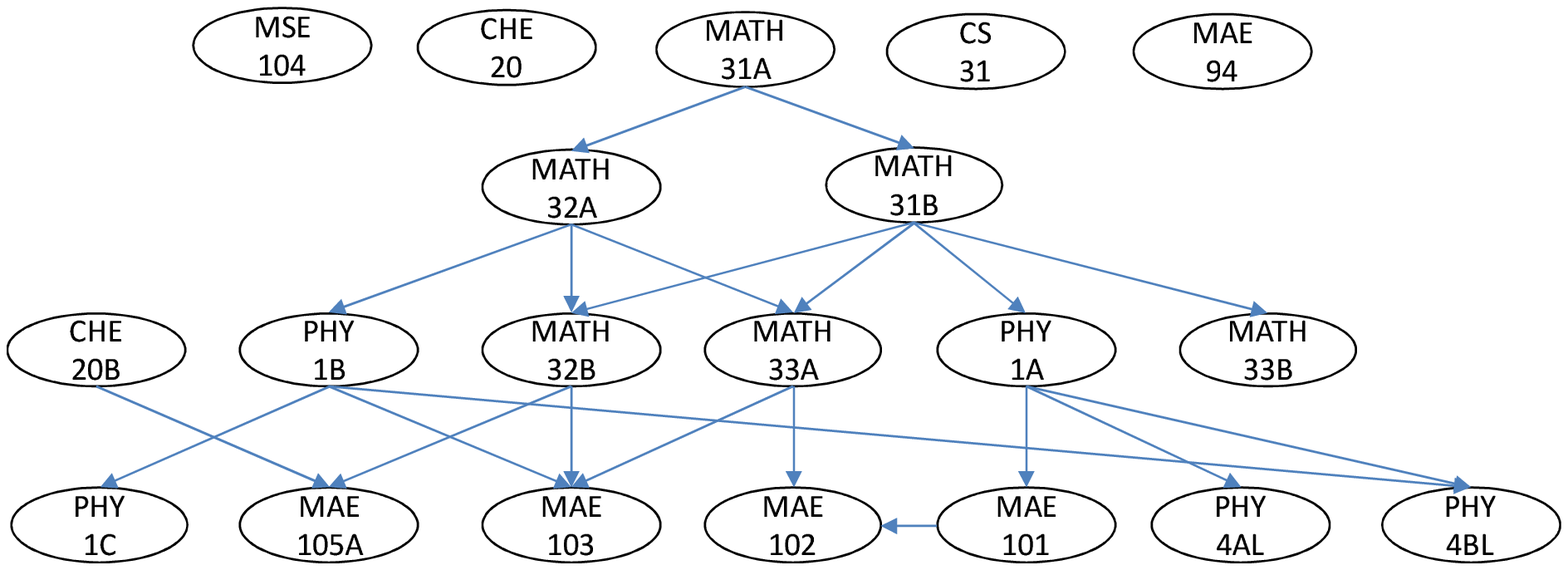}\\
  \caption{The prerequisite graph for the undergraduate program in the Mechanical and Aerospace Engineering department at UCLA. }\label{Prerequisite}
\end{figure}

Given these constraints, the feasible set of courses that a student can take in quarter $t$ given his/her course state $s(t-1)$ can be computed as follows:
\begin{align}
\mathcal{F}(t,s(t-1)) = \{n: &s_n(t-1) = 0; n \in \Gamma(t); \nonumber\\
&\forall m \in P(n), s_m(t-1) = 1\}
\end{align}
Since a student cannot take more than $C$ courses per quarter, the possible combinations of courses that can be elected by a student is
\begin{align}
\mathcal{A}(t, s(t-1)) = \{A: A \subseteq \mathcal{F}(t, s(t-1)); |A|\leq C\}
\end{align}

At the end of each quarter $t$, the student either passes or fails the course that he/she takes in this quarter. The probability that a student fails a course depends on the difficulty of the course as well as how many courses that he/she is taking simultaneously in the same quarter, which can be estimated from the student academic record dataset. Denote the probability that the student fails a course $n$ by $\epsilon_n(k)$ where $k$ is the number of simultaneous courses. Typically, $\epsilon_n(k)$ is a non-decreasing function in $k$ to capture the fact that the student's effort has to be distributed into multiple courses. Depending on the course performance outcome in this quarter, the course state will evolve from $s(t-1)$ to $s(t)$. If the student passes a course $n \in A(t)$, then $s_n(t) = 1$; otherwise, $s_n(t)$ remains 0.

A student graduates when he/she has taken and passed all mandatory courses and at least $E_0 \leq E$ elective courses before the end of $T$ quarters where $E_0$ is a predefined number by the program. The course states in which the student can graduate are called \textrm{terminal states}, which must satisfy $\forall n = 1,...,M, s_n = 1$ and $\sum\limits_{n=M+1}^N s_n \geq E_0$. Let $\hat{S}$ be the set of all terminal course states. There is a reward function $U: \hat{S}\times \{1,...,T\} \to \mathds{R}$ for each terminal state indicating the reward of reaching the terminal state by a specific quarter. For example, $U(\hat{s},t) = 1, \forall \hat{s}\in \hat{S},\forall t$ assigns equal value to all terminal states if the system only cares about whether the student can graduate on time. For another example, $U(\hat{s}, t) = T-t+1, \forall \hat{s}\in \hat{S}$ allows the system to take into account the exact time of graduation.

A course sequence recommendation policy specifies for each course state in any quarter, the next courses that should be taken. Let $\pi(s,t)$ denote the courses that are recommended to take in quarter $t$ given the course state $s$. Given a course sequence recommendation policy $\pi$, starting with any state $s$ in any quarter $t$, the course state $s$ evolves stochastically (since a student may pass or fail the course with probabilities), thereby inducing a probability distribution over the terminal state that can be reached. Let $V(s, t) = \sum\limits_{\hat{s},\tau\geq t} p^\pi_{\hat{s},\tau}(s,t) U(\hat{s}, \tau)$ denote the value of state $s$ in quarter $t$ when policy $\pi$ is adopted where $p^\pi_{\hat{s},\tau}(s,t)$ is the probability of reaching a terminal state $\hat{s}$ in quarter $\tau\geq t$ starting with state $s$ in quarter $t$. The objective of the system is to determine the optimal policy that maximizes the value of the initial state $s(0)$, i.e. $\pi^* = \arg\max_\pi V(s(0),1)$.

Our solution to find the optimal policy $\pi^*$ consists of two phases. In the first phase, we perform a forward search starting from quarter 1 through quarter $T$ to determine all possible course states that can emerge on the learning path. The purpose of this phase is to reduce the course state space in each quarter that the system should look at. In the second phase, we perform a backward induction starting from quarter $T$ through quarter $1$ to compute the optimal set of courses that should be taken in each possible course state. The purpose of this phase is to determine the course sequence recommendation that minimizes the graduation time (or ensure that students graduate before a desired time).

\subsection{Forward Search Phase}
Since the number of possible course states grows exponentially with the number of courses in the curriculum, the course state space can be huge for even a moderate number of courses. However, thanks to the course prerequisite constraint and the course availability constraint, the number of course states that can emerge in a particular quarter can be significantly limited. The purpose of the forward search is to determine the possible course states, thereby reducing the problem complexity.

Let $\mathcal{L}(t)$ denote the set of possible course states by the end of quarter $t$ and $\mathcal{H}(t)$ denote the set of state-course pairs in quarter $t$. Initially $\mathcal{L}(0) = \{s(0)\}$. In each quarter $t$, the algorithm examines each non-terminal course state $s(t-1)\in \mathcal{L}(t-1)$ and determines the feasible course set $\mathcal{F}(t, s(t-1))$  for this course state hence the possible combinations of course $\mathcal{A}(t, s(t-1))$ that can be elected in this quarter. For each combination of courses $A\in \mathcal{A}(t, s(t-1))$, the state-course pair $(s(t-1), A)$ is inserted into $\mathcal{H}(t)$. Then all possible new course states $s(t)$ with respect to $(s(t-1), A)$ is included in $\mathcal{L}(t)$. Moreover, the probability that $s(t-1)$ transits to $s(t)$ is computed by
\begin{align}\label{transition}
&p(s(t)|s(t-1), A)  \nonumber\\
=&\prod_{n: n\in A, s_n(t) =1}(1-\epsilon_n(|A|))\prod_{n:n\in A, s_n(t) = 0}\epsilon_n(|A|)
\end{align}
This algorithm is summarized in Algorithm 1 and is illustrated in Figure 2.

\begin{figure}
  \centering
  \includegraphics[scale = 0.4]{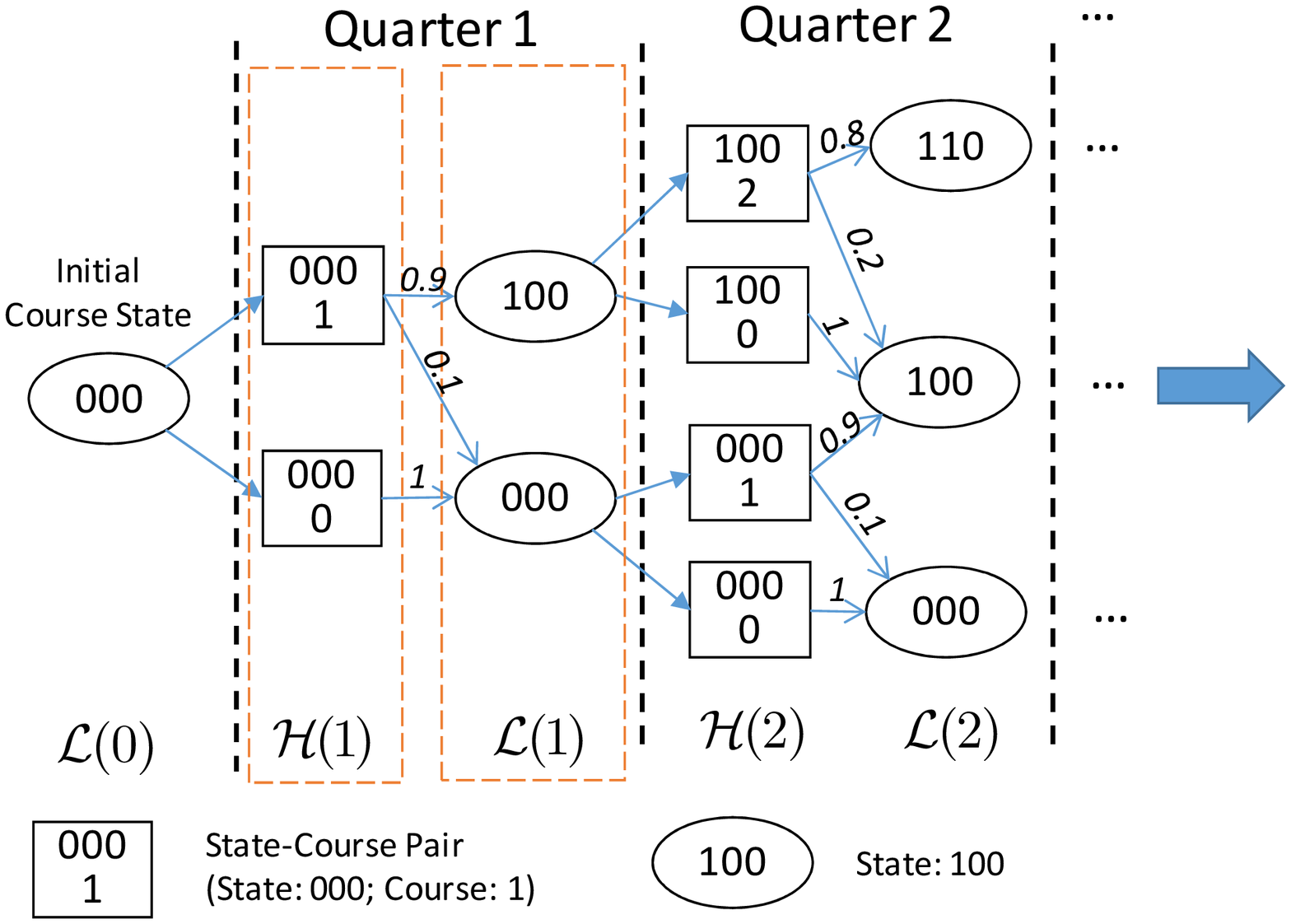}\\
  \caption{Illustration for Forward Search. Each course state (circle) represents the completion status of the three courses. For instance, 100 means that only the first course is taken and passed. Each state-course pair (rectangle) represents the next courses elected in a given state. For instance, 100/2 means that course 2 is elected as the next course to take in a state 100.}\label{ForwardSearch}
\end{figure}

\begin{algorithm}
\caption{Forward Search}
\begin{algorithmic}[1]
\State \textbf{Initialization}: $\mathcal{L}(t) = \emptyset, \mathcal{H}(t) = \emptyset,\forall t$.
\State Initial possible course states $\mathcal{L}(0) = s(0) = \{s_n = 0, \forall n\}$
\State \textbf{For} quarter $t=1$ \textbf{To} quarter $T$:
\State \indent \textbf{For} each course state $s\in \mathcal{L}(t-1)$
\State \indent\indent Determine feasible course set $\mathcal{F}(t, s)$ and feasible courses-to-take set $\mathcal{A}(t, s)$
\State \indent\indent \textbf{For} each feasible combinations of courses $A \in \mathcal{A}(t,s)$
\State \indent\indent\indent Update the current list $\mathcal{H}(t) \leftarrow \mathcal{H}(t)\cup (s, A)$
\State \indent\indent\indent Update $\mathcal{L}(t)$ by adding all possible states
\State \indent\indent \textbf{End For}
\State \indent \textbf{End For}
\State \textbf{End For}
\end{algorithmic}
\end{algorithm}

Next, we analyze the property of the set $\mathcal{L}(t)$ of possible course states.
\begin{lemma}
For any course prerequisite and course availability constraints, $\mathcal{L}(t)$ is weakly expanding, i.e. $\mathcal{L}(0)\subseteq \mathcal{L}(1) ... \subseteq \mathcal{L}(T)$.
\end{lemma}
\begin{proof}
Consider any course state $s \in \mathcal{L}(t)$. Since the student can take no course in the subsequent quarter $t+1$ and hence, the course state remains the same. Therefore, $\mathcal{L}(t+1)$ must at least include $s$ and hence, $\mathcal{L}(t) \subseteq \mathcal{L}(t+1)$.
\end{proof}
Lemma 1 states an intuitive result that the possible course states grow over quarters. However, how fast this set grows depend on the specific course prerequisite dependency as well as the course availability. A formal characterization for a general DAG seems extremely complicated; in the proposition below, we determine the growth rate of $\mathcal{L}(t)$ for two specific cases.

\begin{proposition}
Suppose all courses are offered in all quarters and $C = 1$. (1) If the course prerequisite DAG is a line, then $|\mathcal{L}(t)| = t$. (2) If the course prerequisite DAG is an empty graph (i.e. no prerequisite dependency), then $|\mathcal{L}(t)| = \sum\limits_{\tau=0}^t {N \choose \tau}$. (3) For any general prerequisite DAG, $t \leq |\mathcal{L}(t)| \leq \sum\limits_{\tau=0}^t {N \choose \tau}$.
\end{proposition}
\begin{proof}
(1) By any quarter $t$, the student can take and pass at most $t$ courses. The only possible course states that can emerge in quarter $t$ is $\{1\}, \{1,2\},..., \{1,2,...,t\}$. Therefore, there are totally $t$ possible course states.

(2) In quarter 1, the student can take any one of the $N$ courses or does not take any course. The number of possible states is thus $N + 1$. If this student passes this course, he/she can take any one of the remaining $N-1$ courses or does not take any course in quarter 2. The possible states with two courses passed is ${N \choose 2} + N + 1$. Continuing in this way, we obtain the result.

(3) It is easy to see that (1) and (2) are two extreme cases of a general DAG, therefore, the number of possible states is bounded by $t$ and $\sum\limits_{\tau=0}^t {N \choose \tau}$.
\end{proof}

Proposition 1 shows that the possible course states highly depend on the course prerequisite constraints. The number of possible course states grows linearly when the prerequisite is strict while it grows exponentially when the prerequisite is loose. Note that in the above proposition we did not impose any restriction on the course availability. In practice, only a limited number of courses are offered in each quarter, thereby further limiting the size of $|\mathcal{L}(t)|$.

\textbf{Remark}: $\mathcal{L}(t)$ contains all possible course completion states by the end of quarter $t$ if students randomly choose their course sequences. Therefore, $\mathcal{L}(t)$ reflects how diverse the students learning experience can be by adopting a specific curriculum. The larger $\mathcal{L}(t)$, the more diverse learning experience that students can have by the end of quarter $t$. A more diverse student learning experience (which is determined by the curriculum but not the course sequence recommendation) has two implications. On one hand, finding the best course sequence to recommend becomes more difficult since the searching space is larger. On the other hand, the best course sequence may yield better learning outcome.

\subsection{Backward Induction Phase}
The outcome of the Forward Search phase is actually an AND/OR graph where each course subsequence is a subgraph of the AND/OR graph. In this graph, each OR node represents a course state $s \in \mathcal{L}(t)$ in which different possible combinations of subsequence courses can be elected. Each AND node $(s, A)$ corresponds to electing courses $A$ in course state $s$. The AND node also stores a probability distribution over the possible next states by taking these courses, which is computed by \eqref{transition}. The value of the optimal course sequence recommendation for this AND/OR graph can be computed by a bottom-up sweep through the graph. This computation can be viewed as a backward induction. First, the value of all OR nodes that are non-terminal states and all AND nodes are initialized to 0, and the value of all OR nodes that are terminal states are initialized to the reward of the corresponding terminal states. Next, starting from quarter $T$, for each quarter $t$, we update the value of the AND nodes in $\mathcal{H}(t)$ using the value of the OR nodes in $\mathcal{L}(t)$, i.e. $\forall (s,A) \in \mathcal{H}(t)$,
\begin{align}
Q(s, t, A) = \sum\limits_{s'\in \mathcal{L}(t)} p(s'|s,A) V(s', t)
\end{align}
We then update the value of the OR nodes in $\mathcal{L}(t-1)$ using the value of the AND nodes in $\mathcal{H}(t)$, i.e. $\forall s \in \mathcal{L}(t-1)$,
\begin{align}
V(s, t-1) = \max_{A} Q(s, t, A)
\end{align}
and we record the combinations of courses in the course recommendation policy
\begin{align}
\pi^*(s, t-1) = \arg\max _A Q(s, t, A)
\end{align}
As mentioned, depending on the choice of the reward function for the terminal course states, different system objectives can be achieved by solving the above problem. The algorithm is summarized in Algorithm 2 and illustrated in Figure \ref{BackInduction}.

\begin{figure}
  \centering
  \includegraphics[scale = 0.4]{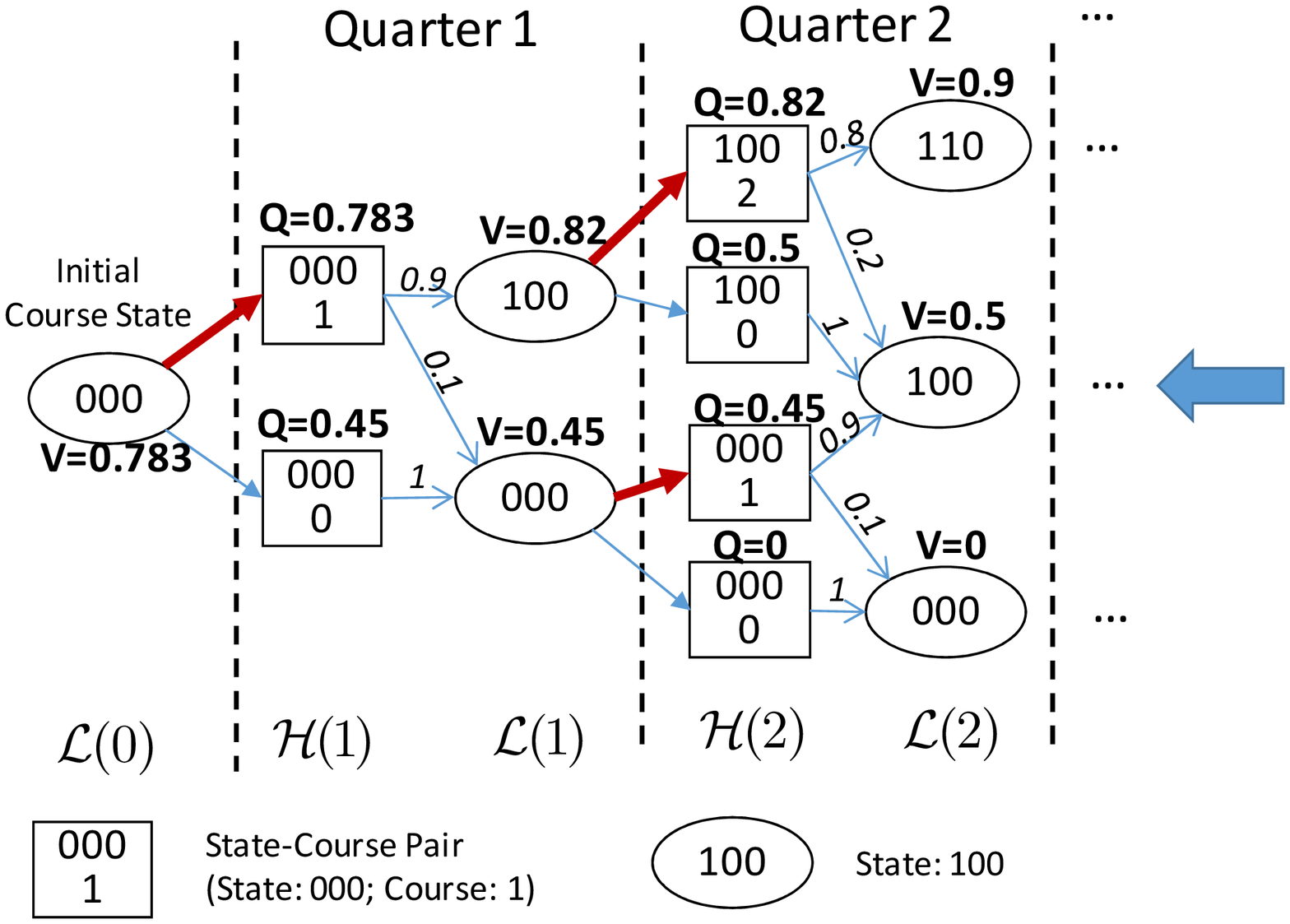}\\
  \caption{Illustration for Backward Induction. The red thick arrow represents the next course elected for each course state. }\label{BackInduction}
\end{figure}

\begin{algorithm}
\caption{Backward Induction}
\begin{algorithmic}[1]
\State \textbf{Initialization}: $\forall s \not\in \hat{s}$, $V(s,t) = 0, Q(s, t, A) = 0$; $\forall s \in \hat{s}$, $V(s,t) = U(s,t)$.
\State \textbf{For} quarter $t=T$ \textbf{To} quarter $1$
\State \indent \textbf{For} each state and courses-to-take pair $s, A\in \mathcal{H}(t)$
\State \indent\indent Update $Q(s, t, A) = \sum\limits_{s'\in\mathcal{L}(t)}p(s'|s,A) V(s', t)$
\State \indent \textbf{End For}
\State \indent \textbf{For} each course state $s \in \mathcal{L}(t-1)$
\State \indent\indent Update value function $V(s,t-1) = \max_A Q(s,t,A)$
\State \indent\indent Update policy $\pi(s,t-1) = \arg\max_A Q(s, t, A)$
\State \indent \textbf{End For}
\State \textbf{End For}
\end{algorithmic}
\end{algorithm}

Next, we analyze the property of the value function $V(s,t)$ and the structure of the optimal policy $\pi^*$. We say $s\prec \tilde{s}$ if $s_n = 1$ implies $\tilde{s}_n = 1$. That is, all courses that are passed in $s$ are also passed in $\tilde{s}$.
\begin{proposition}
For any $t$, $s$ and $\tilde{s}$, if $s\prec \tilde{s}$, then $V(\tilde{s}, t) \geq V(s,t)$.
\end{proposition}
\begin{proof}
Let $A^*$ be the optimal courses for $s$, i.e. $V(s(t),t) = Q(s(t), t+1, A^*)$. To prove $V(\tilde{s}(t), t) \geq V(s(t), t)$, we will prove for any $A^*$, we can find $\tilde{A} \in \mathcal{A}(t+1, \tilde{s}(t))$ such that $Q(\tilde{s}(t), t+1, \tilde{A}) \geq Q(s(t), t+1, A^*)$. Since $V(\tilde{s}(t),t) \geq Q(\tilde{s}(t), t+1, \tilde{A})$, we will get $V(\tilde{s}, t) \geq V(s,t)$. The proof uses induction. It is straightforward to see that $V(\tilde{s}, T) \geq V(s, T)$ since if $s$ is a terminal state, $\tilde{s}$ must also be a terminal state. Suppose the claim holds for all $t+1, ..., T$.

Case 1: $A^* \in \mathcal{A}(t+1, \tilde{s}(t))$. In this case, we let $\tilde{A} = A^*$, namely we select the same set of courses for the student to take. Since the courses taken are the same, the probability to pass each of these courses is the same. For each possible next state $s'(t+1)$, there must exist a corresponding next state $\tilde{s}'(t+1)$ such that $p(s'(t+1)|s(t),A^*) = p(\tilde{s}'(t+1)|\tilde{s}(t),A^*)$. Moreover, $s'(t+1) \prec \tilde{s}'(t+1)$. By induction, we have $V(s'(t+1), t+1) \leq V(\tilde{s}'(t+1), t+1)$. Therefore,
\begin{align}
&Q(s(t), t+1, A^*) \nonumber\\
=& \sum\limits_{s'(t+1)} p(s'(t+1)|s(t+1),A^*) V(s'(t+1), t+1)\nonumber\\
=&\sum\limits_{s'(t+1)} p(\tilde{s}'(t+1)|\tilde{s}(t),A^*) V(s'(t+1), t+1)\nonumber\\
\leq &\sum\limits_{\tilde{s}'(t+1)} p(\tilde{s}'(t+1)|\tilde{s}(t),A^*) V(\tilde{s}'(t+1), t+1) \nonumber\\
= &Q(\tilde{s}(t), t+1, A^*)
\end{align}

Case 2: $A^* \not\in \mathcal{A}(t+1, \tilde{s}(t))$. In this case, we let $\tilde{A}$ to be the largest subset of $A^*$ that belongs to $\mathcal{A}(t+1, \tilde{s}(t))$. Moreover, since $s\prec \tilde{s}$, the remaining subset of $A^*$ must be courses that have already been passed in $\tilde{s}$. Suppose that courses in $A^*-\tilde{A}$ are passed with probability 1 starting from $s(t)$ and the remaining courses are passed with probability $1-\epsilon_n(|\tilde{A}|) > 1-\epsilon_n(|A^*|)$. Due to induction, such relaxation provides an upper bound on $Q(s(t), t+1, A^*)$. Moreover, this relaxation is the same as
\begin{align}
&\sum\limits_{\tilde{s}'(t+1)} p(\tilde{s}'(t+1)|\tilde{s}(t),A^*) V(\tilde{s}'(t+1), t+1) \nonumber\\
=& Q(\tilde{s}(t), t+1, \tilde{A})
\end{align}
Thus, $Q(s(t), t+1, A^*) \leq Q(\tilde{s}(t), t+1, \tilde{A})$.
\end{proof}
Proposition 2 implies that a course state where a larger set of courses have been passed has a higher value since there is more flexibility in choosing subsequent course sequences. A question naturally arises that is it always better to take as many courses as possible in any quarter? The answer turns out to be correct only in certain scenarios. The following proposition identifies one of such scenarios.

\begin{proposition}
Suppose $\mathcal{E} = \emptyset$ (i.e. no prerequisites for all courses) and $\Gamma(t) = \mathcal{N},\forall t$ (i.e. each course is offered in all quarters), if $\epsilon_n(|A|) = \epsilon_n$ is a constant, then at any $t$ and given state $s(t-1)$, the optimal policy recommends that the student should take the maximum number of $C$ courses.
\end{proposition}
\begin{proof}
Consider any policy that selects $K < C$ courses in quarter $t$ and a policy that selects these $K$ courses plus one more course $n$. For any possible next course state $s'$, by selecting one more course, we have
\begin{align}
p(s'|s,K) = p(s'|s, K+1) + p(\tilde{s}'|s, K+1)
\end{align}
where $s'\prec \hat{s}'$ and $\hat{s}'_n = 1$. Due to the course failure probability being independent, $p(s'|s, C) = \epsilon_n p(s'|s,K)$ and $p(\tilde{s}'|s, C) = (1-\epsilon_n)p(s'|s,K)$. Since $V(\tilde{s}') \geq V(\tilde{s})$ according to Proposition 2, we have $Q(s,K) \leq Q(s, K+1)$. Therefore, taking $C$ courses yields higher value than taking fewer courses.
\end{proof}

Proposition 3 states that it is always better to take many courses as early as possible when there is no constraint on course prerequisite and course availability. In practice, course sequence recommendations are subject to many constraints. We provide a counter-example below to show that the result of proposition 3 does not hold in the general case.

\textbf{\textit{Counter-Example}}: Consider a program consisting of two courses and two quarters. Thus $\mathcal{N} = \{1,2\}$ and $T = 2$. Assume that there is no prerequisite course dependency, i.e. $\mathcal{E} = \emptyset$. The course availability is $\Gamma(1) = \{1,2\}$ and $\Gamma(2) = \{2\}$. That is, course 1 is offered only in the first quarter. $C = 2$ so students are allowed to take up to 2 courses. Let $\epsilon(K)$ denote the probability of failing a course when $K$ courses are taken simultaneously.
\begin{itemize}
  \item Option 1: Take 2 courses in quarter 1. The probability of graduation on time can be computed as $(1-\epsilon(2))(1-\epsilon(2) + \epsilon(2)(1-\epsilon(1))) = (1-\epsilon(2))(1-\epsilon(2)\epsilon(1))$
  \item Option 2: Take course 1 only in quarter 1. The probability of graduation on time can be computed as $(1-\epsilon(1))(1-\epsilon(1))$.
\end{itemize}
Thus if $(1-\epsilon(2))(1-\epsilon(2)\epsilon(1)) < (1-\epsilon(1))(1-\epsilon(1))$, then taking only 1 course in the first quarter leads to a higher probability of graduation. For instance, taking $\epsilon(1) = 0.1$ and $\epsilon(2) = 0.2$ satisfies this condition. This counter-example demonstrates the need for carefully planning the course sequence according to the course prerequisite and course availability because myopic course selection may lead to lower learning reward if these constraints are ignored.

\textbf{Remark}: The complexity of the proposed forward-search backward-induction algorithm to determine the optimal policy depends on, the number of courses, the specific course prerequisite and availability constraints. As shown in Proposition 1, the set of possible course states grow at different speeds depending on these constraints. The time and memory complexity generally depend on the size of $\mathcal{L}(t)$ and $\mathcal{H}(t)$, namely $O(\sum\limits_{t=0}^T(|\mathcal{L}(t)| + |\mathcal{H}(t)|))$, which can be large in certain scenarios. However, since our algorithm in this section is an offline algorithm and only needs to be executed once, complexity is not a big concern.

\subsection{Implications on Curriculum Planning}
Teaching resources are limited. An important question for curriculum planning is how to allocate the limited available teaching resources to the courses to minimize students' time-to-graduation. Our framework can be helpful in answering part of this question. The proposition below shows, in a simplified scenario, that courses with many dependent courses should receive more teaching resource in order to minimize the time-to-graduation.

\begin{proposition}
Consider $N+1$ courses and the prerequisite DAG satisfies $P(n) = \{1\}, \forall n > 1$ (i.e. course 1 is the prerequisite course of all other courses). Consider the following two cases of course availability constraints
\begin{itemize}
  \item Case 1: Course 1 is offered in each quarter with probability $p < 1$ and all other courses is offered in each quarter with probability 1.
  \item Case 2: Course 2 is offered in each quarter with probability $p < 1$ and all other courses is offered in each quarter with probability 1.
\end{itemize}
Assume $C = 1$ and if $p < \frac{1}{\sqrt{N} + 1}$, then the expected graduation time in case 1 is larger than in case 2.
\end{proposition}
\begin{proof}
Since course 1 is the prerequisite of all the other courses, it has to be passed before any other course can be taken.

Case 1: The expected time to finish course 1 is $\frac{1}{(1-\epsilon)p}$ where $\epsilon$ is the probability of failing a course. The expected time to finish the remaining $N$ courses is $\frac{N}{1-\epsilon}$. Thus, the total expected time to graduate is
\begin{align}
\tau_1 = \frac{1}{(1-\epsilon)p} + \frac{N}{1-\epsilon}
\end{align}

Case 2: The expected time to finish course 1 is $\frac{1}{(1-\epsilon)}$. Computing the expected time to finish the remaining $N$ courses is more complicated since one of the courses has different course availability than others. We consider an upper bound on this time. In particular, this upper bound is
\begin{align}
\tau_2 < \frac{1}{1-\epsilon} + \max\{\frac{N-1}{(1-\epsilon)(1-p)}, \frac{1}{(1-\epsilon)p}\} \triangleq \bar{\tau}_2
\end{align}
Depending on the values of $p$ and $N$, $\bar{\tau}$ can be written as
\begin{equation}
\bar{\tau}_2 = \left\{
\begin{array}{ll}
\frac{1}{1-\epsilon} + \frac{1}{(1-\epsilon)p}, &\textrm{if } p\leq\frac{1}{N}\\
\frac{1}{1-\epsilon} + \frac{N-1}{(1-\epsilon)(1-p)}, &\textrm{if } p > \frac{1}{N}
\end{array}\right.
\end{equation}
If $p \leq \frac{1}{N}$, it is easy to see that $\bar{\tau}_2 < \tau_1$ and hence $\tau_2 < \tau_1$. If $p > \frac{1}{N}$,
\begin{align}
&\tau_1 - \bar{\tau}_2 = \frac{1}{1-\epsilon}\left(\frac{1}{p} + N - 1 - \frac{N-1}{1-p} \right)\nonumber\\
=&\frac{1}{1-\epsilon}\frac{(1-p)^2 - p^2 N}{p(1-p)}
\end{align}
Therefore, if $p < \frac{1}{\sqrt{N} + 1}$, then $\tau_1 > \bar{\tau}_2 > \tau_2$.
\end{proof}
It is intuitive to understand that courses that are prerequisites for many other courses are more ``important''. Subsequent courses can be taken only if the prerequisite course is passed. Therefore, much time will be wasted if students cannot take the prerequisite courses. On the other hand, even if a student needs to but cannot take a course that is not a prerequisite for other courses, he/she can still take other courses while waiting for this course to become available.

\subsection{Joint Optimization of Time-to-Graduation and GPA Performance}
So far, we focused on constructing course sequence recommendation policies that minimize the time-to-graduation. However, the exact learning performance upon graduation (e.g. GPA) is neglected. Nevertheless, the above dynamic programming based framework can be easily extended to the case of joint optimization of time-to-graduation and GPA performance, provided that a sufficiently large dataset is available to estimate the various model parameters. We elaborate on this point below.

The grade that a student can receive in a course often depends on the grades that the student received in the prerequisite courses and perhaps how long ago the prerequisite courses were taken. Thus, the course sequence that the student is taking may have a significant impact on the GPA that he/she can obtain. To account for this effect, we modify the problem formulation: instead of keeping a course completion state, which records the courses that have been taken and passed, we keep a course performance state, which records the grades that the student has received in the passed courses and when these courses were taken. Then for each performance state, the policy tries to find the set of next courses to take in order to maximize an objective function that jointly considers the time of graduation and the obtained final GPA. However,  solving this problem requires addressing several key challenges. First, the course performance state space is significantly larger than the course completion state space and grows exponentially with the number of possible grades. In particular, suppose the number of grade levels is $K$, then the number of all possible states is $K^N$. Second, a huge dataset of student records is needed to estimate the conditional probabilities of grades of each course depending on all possible course performance states. Therefore, solving the optimal course recommendation policy is extremely difficult.

To derive efficient solutions without a large initial dataset, we construct course sequence recommendation policies that jointly consider the GPA and time to graduation in two steps. Firstly, we determine a set of course sequence recommendation policies that satisfy desired time to graduation constraints using our method developed in this section. Since there is a lack of dataset to estimate the course failure probabilities, the course sequence recommendation policy can even be constructed by ignoring the course failure probabilities (i.e. treating $\epsilon_n \to 0, \forall n$). In the second step, we maximize the student GPA by considering only the policies derived in the first step. In particular, we aim to select for each student a personalized policy that most suits this student and results in the highest GPA. We formalize the personalized policy selection problem in the next section.

\textbf{Remark}: Readers may wonder why there are multiple solutions of course sequence recommendation policies from the first step so that personalization is possible in the second step. A couple of reasons can lead to multiple solutions. First, multiple solutions may occur due to ties, which are more likely to happen when the randomness disappears as $\epsilon_n \to 0$. For instance, consider 4 courses $\{1,2,3,4\}$ and course 1 is the prerequisite of the other three courses. Assume $C = 2$ and all courses are offered in all quarters, then course sequences $1 \to \{2,3\} \to 4$, $1\to \{2, 4\} \to 3$ and $1\to \{3, 4\} \to 2$ all lead to the same shortest time-to-graduation. Second, instead of keeping only the course sequence recommendation policy that yields the shortest time-to-graduation, the first step of our approach can also generate a set of course sequence recommendation policies that result in on-time graduation.

\section{Online Recommendation Policy Personalization}

\subsection{Problem Formulation}
We consider an online setting where students enter the program in sequence. The students are indexed by $\{1,2,...,i,...\}$. Students come with different background (e.g. schools from which the students graduated, SAT scores). We use a context vector $\theta_i \in \Theta$ to denote the student background where $\Theta = [0,1]^W$ is the normalized context space with dimension $W$. We have a set of $Z$ course sequence recommendation policies constructed, denoted by $\mathcal{Z}$, using our method proposed in Section III. These recommendation policies ensure that students will graduate early with high probability. However, the impact of these recommendation policies on the students' GPA performance is unknown a priori and may be different for students with different backgrounds.

For each student $i$, the system selects one of the $Z$ policies to recommended course sequence to this student. When the student completes the program by following the recommended course sequence, the GPA that he/she obtains is revealed as $r_i$. Let $\mu_z(\theta) = \mathbb{E}\{r|\theta\}$ be the expected GPA for the student with background $\theta$ if a recommendation policy $z$ is adopted. If $\mu_z(\theta)$ were known for each policy $z$, then the policy selection problem would have been simple - selecting $z^*(\theta) = \arg\max_{z} \mu_z(\theta)$ maximizes the expected GPA for this student. However, since the effectiveness of the recommendation policies is unknown a priori, the best policies must be learned for each student.

Let $\sigma$ be an online learning algorithm for policy selection and $\sigma(i) \in \mathcal{Z}$ denote the policy that is used on student $i$. We use learning regret as the performance metric for a learning algorithm. The learning regret up to student $I$ is defined as the aggregate GPA difference between our learning algorithm and the oracle solution that selects the best policy $z^*(\theta_i), \forall i$, i.e.
\begin{align}
\textrm{Reg}(I) := \mathbb{E}[\sum\limits_{i=1}^I \mu_{z^*(\theta_i)}(\theta_i) - \sum\limits_{i=1}^I r_i(\sigma(i))]
\end{align}
where the expectation is taken with respect to the randomness in grade realization and the selected policies. The regret characterizes the loss incurred due to the unknown system dynamics and gives convergence rate of the total expected GPA of the learning algorithm to the value of the oracle solution. The regret is non-decreasing in the total number of incoming students, but we want it to increase as slow as possible. Any algorithm whose regret is sublinear in $I$, i.e. $\textrm{Reg}(I) = O(I^\alpha)$ such that $\alpha < 1$, will converge to the optimal solution in terms of the average reward, i.e. $\lim\limits_{I\to\infty}\frac{\textrm{Reg}(I)}{I} = 0$. The regret of learning also gives a measure for the rate of learning. A smaller $\alpha$ will result in a faster convergence to the optimal average reward and thus, learning the optimal course sequence recommendation is faster if $\alpha$ is smaller.

\subsection{Context-Aware Adaptive Policy Selection}
A natural way to learn a course sequence recommendation policy's effectiveness is to record and update the sample mean GPA obtained as students arrive and complete the program by adopting this policy. Using such a sample mean-based approach for policy selection is the basic idea of our learning algorithm. However, major challenges still remain. Without using the context information, we have only learned the average performance of each recommendation policy and thus, a single policy will always be selected. On the other hand, personalizing the policy for each student according to his/her background can be very difficult since the students can have diverse background and hence the context space $\Theta$ can be huge. The sample mean reward approach can fail to work since there will be very limited number of students who have the same background. Our method to overcome this difficulty is by exploiting the similarity of students based on the assumption that students with similar background will achieve similar \textit{expected} GPA by following the same course sequence recommendation policy. Our learning algorithm starts with a larger context space to learn the best recommendation policy for this space and then gradually refines the learning by partitioning the context space into smaller spaces.

Before we describe the details of our algorithm, we introduce several useful concepts.
\begin{itemize}
  \item \textbf{Student Cluster}. A student cluster is represented by the range of context information that is associated with students in the cluster. In this paper, we consider student clusters that are created by uniformly partitioning the context space on each dimension, which are enough to guarantee sublinear learning regrets. Thus, each student cluster is a $W$-dimensional hypercube with side length being $2^{-l}$ for some $l$. This hypercube represents a level-$l$ student cluster. At any moment in time when a recommendation policy is applied to a student $i$, the algorithm keeps a set of mutually exclusive student clusters that cover the entire student population. We call these student clusters the active student clusters, and denote this set by $\Omega$. Since the active student clusters evolve (i.e. become more refined) as more students are enrolled and graduate, the active set $\Omega^i$ uses a superscript $i$, which is the student index, to represent its dynamic nature. For instance, in the one-dimensional case, $\{[0, 1/2), [1/2,1]\}$ is a feasible set of active student clusters and $\{[0, 1/4), [1/4, 1/2), [1/2, 3/4), [3/4, 7/8), [7/8, 1]\}$ is another feasible set of active student clusters. Figure \ref{Clustering} illustrates a 2-dimensional student clustering.
  \item \textbf{Counters}. For each active student cluster $\mathcal{C}$, the algorithm maintains $Z$ counters: for each recommendation policy $z \in \mathcal{Z}$, $M_\mathcal{C}(z)$ records the number of students so far in which $z$ is applied to.
  \item \textbf{GPA Estimates}. For each active student cluster, the algorithm also maintains the sample mean GPA estimates $\bar{r}_\mathcal{C}(z)$ for each policy $z \in \mathcal{Z}$ using the realized GPA of students that belong to $\mathcal{C}$ so far.
\end{itemize}

\begin{figure}
  \centering
  \includegraphics[scale = 0.4]{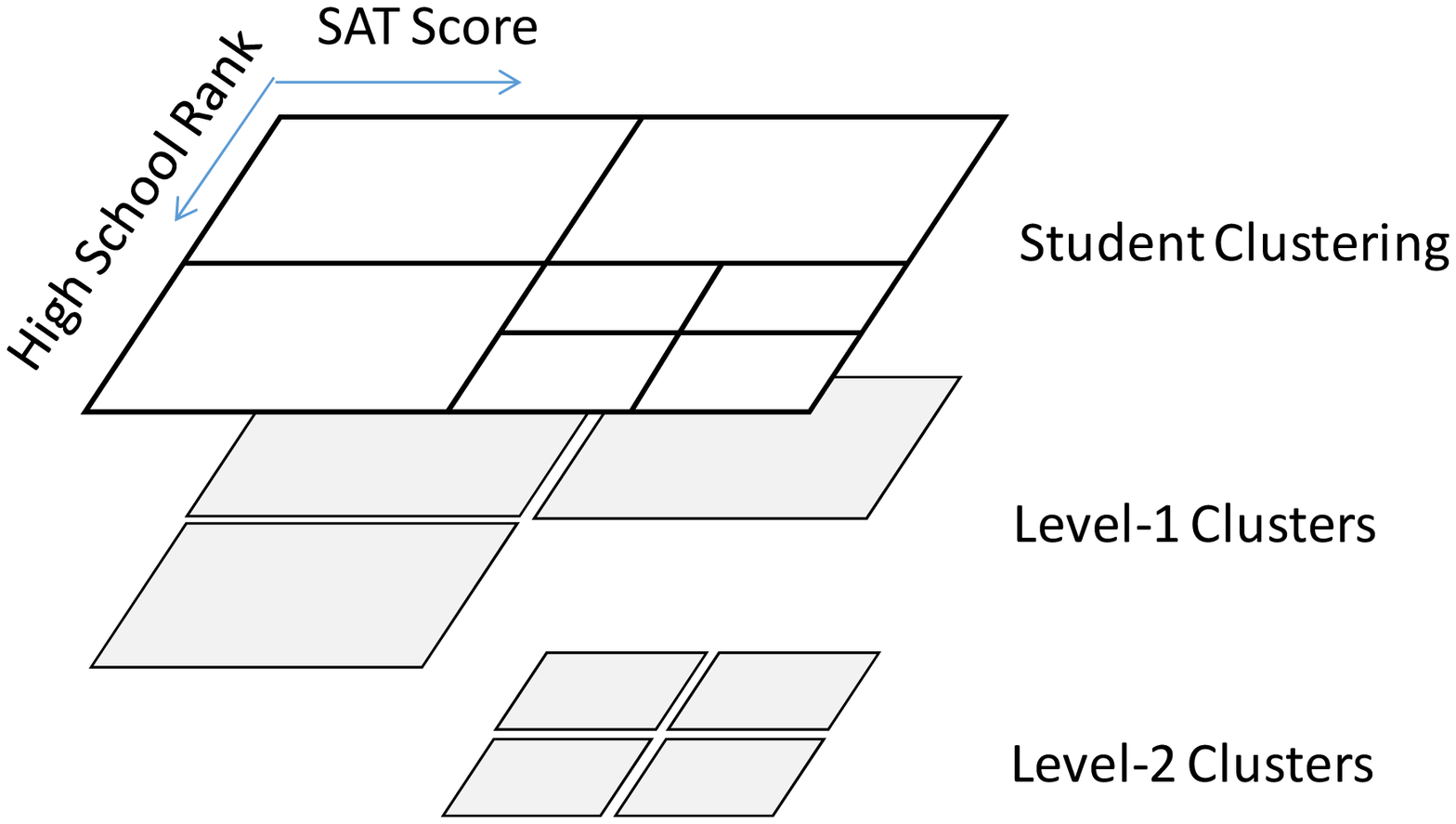}\\
  \caption{Illustration of 2-D Student Clusters.}\label{Clustering}
\end{figure}

The algorithm (see Algorithm 3) works as follows. For each student $i$, the algorithm works in two steps.
\begin{itemize}
  \item \textbf{Policy Selection Step}. When an student arrives, the algorithm first checks which active student cluster $\mathcal{C}\in \Omega^i$ it belongs to. Then it investigates counter $M_\mathcal{C}(z)$ for all $z \in \mathcal{Z}$ to see if there exists any under-explored policy such that $M_\mathcal{C}(z) \leq \gamma(i, l)$ where $\gamma(i, l)$ is a deterministic control function depending on the index of the student $i$ and the level of the student cluster $l$. If there exists such an under-explored policy $z$, then the algorithm uses this policy to recommend course sequence for this student $i$. If there does not exist any under-explored policy, then the algorithm selects the policy with the highest GPA estimate for the cluster $\mathcal{C}$, i.e. $\arg\max_{z}\bar{r}_\mathcal{C}(z)$.
  \item \textbf{Variable Update Step}. After the student completes the program and the GPA is realized, the GPA estimate of the selected policy is updated. 
      Moreover, if the number of students in the student cluster $\mathcal{C}$ satisfies $\sum\limits_{z} M_\mathcal{C}(z) \geq \zeta(l)$ where $\zeta(l)$ is a deterministic control function depending on the level of student cluster, the current student cluster $\mathcal{C}$ is partitioned into $2^W$ level-$(l+1)$ smaller student clusters. For the next student on, $\mathcal{C}$ is deactivated and the new level-$(l+1)$ student clusters are activated.
\end{itemize}

\textbf{Remark}: In the policy selection step, the algorithm may select an under-explored policy for a student. The purpose of this exploration is to learn the effectiveness of every policy with high confidence. However, this exploration does raise fairness issues for some students. There are a couple of solutions that can be used to address this ``unfair'' policy selection issue. First, the policy selection is merely a recommendation, students can still freely choose whichever course sequence they want to follow. Second, rewards mechanisms can be designed to incentivize students to follow the recommended policy. For example, students who follow an under-explored policy can enjoy a lower tuition or receive some form of compensation through some special fellowship.

\begin{algorithm}
\caption{Policy Selection and Adaptive Clustering}
\begin{algorithmic}[1]
\State Initialize $\Omega = \Theta$, $\bar{r}_{\Theta}(z) = 0, M_{\Theta}(\pi) = 0, \forall z\in \mathcal{Z}$.
\For {each student $i$}
\State Determine active cluster $\mathcal{C}\in \Omega^i$ such that $\theta^i\in \mathcal{C}$
\State \textbf{Case 1}: $\exists z \in \mathcal{Z}$ such that $M_\mathcal{C}(z) \leq \gamma(i,l)$
\State ~~~Randomly select among such policies $\sigma^i = z$
\State \textbf{Case 2}: $\forall z \in \mathcal{Z}$, $M_\mathcal{C}(z) > \gamma(i,l)$
\State ~~~Select $\sigma^t = \arg\min\limits_{z\in\mathcal{Z}} \bar{r}_\mathcal{C}(z)$.
\State Set $M_\mathcal{C}(\sigma^i) \leftarrow M_\mathcal{C}(\sigma^i) + 1$
\State (The student GPA $r^i$ is realized.)
\State Update $\bar{r}_\mathcal{C}(\sigma^t)$
\If{$\sum_z M_\mathcal{C}(z) \geq \zeta(l)$}
\State Uniformly partition $\mathcal{C}$ into $2^W$ level-$(l+1)$ student clusters.
\State Update the set of active clusters $\Omega^i$.
\State Update the counters and GPA estimates for all new student clusters
\EndIf
\State \textbf{endif}
\EndFor
\State \textbf{endfor}
\end{algorithmic}
\end{algorithm}

\subsection{Control Function Determination}
In this subsection, we determine the control function $\gamma(i,l)$ and $\zeta(l)$ and evaluate the performance of the policy selection algorithm. The following assumption on student similarity is needed for the regret analysis but not needed in the algorithm.
\begin{assumption}
(Student Similarity). For each policy $z \in \mathcal{Z}$, there exists $\alpha > 0$ such that for all $\theta, \theta'\in \Theta$, we have $|\mu_\theta(z) - \mu_{\theta'}(z)| \leq \|\theta, \theta'\|^\alpha$.
\end{assumption}
The above assumption states that if the student background (i.e. context) is similar, then the expected GPA by using the same course sequence recommendation policy is also similar.

\begin{proposition}
By setting $\gamma(i,l) = 2^{2\alpha l}\ln i$ and $\zeta(l) = A2^{pl}$, the learning regret for students up to $I$ is
\begin{align}
\textrm{Reg}(I) \leq &\sum\limits_{l=1}^{l_{max}(I)}[S_l(I)Z 2^{2\alpha l}\ln I  \nonumber\\
+&Y_l(I)(Z\sum\limits_{i=1}^\infty i^{-2} + A(2W^{\alpha/2} + 2)2^{(p-\alpha)l})]
\end{align}
where $Y_l(I)$ is the number of level-$l$ student clusters that are activated at student $I$, $S_l(I)$ is the number of level-$l$ active student clusters at student $I$ and $l_{max}(I)$ is the maximum level of student clusters at student $I$.
\end{proposition}
\begin{proof}
We break down the learning regret into three parts $\textrm{Reg}(I) = \textrm{Reg}_1(I) + \textrm{Reg}_2(I) + \textrm{Reg}_3(I)$ that are respective regrets due to exploration, selection of suboptimal policies in exploitation and selection of near-optimal policies in exploitation. For each of these three parts we can provide a bound. In particular, $\textrm{Reg}_1(I)$ can be bounded by $\sum\limits_{l=1}^{l_{max}(I)}[S_l(I)Z 2^{2\alpha l}\ln I]$ since the number of exploration steps increases sublinearly in $I$. $\textrm{Reg}_2(I)$ can be bounded by $Y_l(I)Z\sum\limits_{i=1}^\infty i^{-2}$ since the probability of choosing a suboptimal policy in exploitation steps decreases sufficiently rapidly using a Chernoff-Hoeffding bound. $\textrm{Reg}_3(I)$ can be bounded by $Y_l(I)A(2W^{\alpha/2} + 2)2^{(p-\alpha)l}$ since the marginal cost of selecting a near-optimal policy decreases sufficiently rapidly.
\end{proof}

\begin{corollary}
If the student context arrivals by student $I$ is uniformly distributed, we have
\begin{align}
\textrm{Reg}(I) < I^{\frac{2\alpha + W}{3\alpha + W}}[Z(\ln I + \sum\limits_{i=1}^\infty i^{-2}) + (W^{\alpha /2} + 2)2^{2\alpha + W}]
\end{align}
\end{corollary}
As we can see, the regret bound is sublinear in the number of students $I$ and hence, if $I$ is sufficiently large, then the average regret will be close to 0, which means that the optimal average GPA is achieved.

\section{Experiments}
In this section, we apply the forward-search backward-induction method presented in Section III, and the regret minimization learning algorithm developed in Section IV to our dataset.

\subsection{Dataset}
Our experiments are based on a dataset from the undergraduate curriculum of the Mechanical and Aerospace Engineering (MAE) department at UCLA. The dataset contains the course sequences and the course grades of 1444 anonymized students who graduated between the academic years 2013 and 2015. The course availability varies across years; typically, a course is offered either once or twice every academic year but some courses are offered once every two years. UCLA adopts the quarter system and in each academic year and courses are mostly offered in Fall, Winter, Spring quarters but not Summer quarters. Therefore, we will consider that one academic year consists of three quarters (hence, four academic years equal 12 quarters). 

The dataset also includes context information of the students, including their SAT scores and their high school GPAs. We observe that many students in the dataset take courses outside of the curriculum (such as the art courses). Our model can be extended to capture this possibility with added model complexity but our experiment in this paper does not consider the recommendation of such courses. Some of the students in the dataset are transfer students, and they do not need to take the same number of courses as the regular students, since they may have fulfilled several requirements before coming to UCLA. Since the course information before the transferring is not included in this dataset, we exclude such transfer students from our analysis. Figure \ref{Pie} shows the graduation time distribution of the students in the dataset. As we can see, even though the majority of students graduate on time within the desired four years (12 quarters), many students do not graduate on time and stay in the college for one or even two more years. There is also a noticeable difference in the graduation time distributions for students with different context information: students with higher math SAT scores have a higher probability to graduate on time. This suggests that personalization based on the students' context information has indeed the potential to provide better learning experience and lead to better learning outcomes.

\begin{figure}
  \centering
  \includegraphics[scale = 0.5]{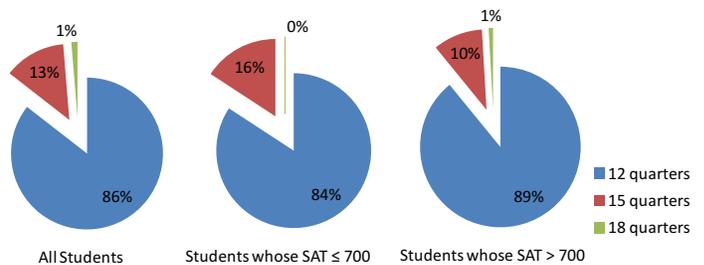}\\
  \caption{Graduation Time Distribution and the Impact of Math SAT score.}\label{Pie}
\end{figure}

\subsection{Impact of constraints}
In this subsection we illustrate how course prerequisite and availability reduce the number of possible course sequence recommendations for students in MAE department at UCLA. Figure \ref{Prerequisite} in Sec. III depicts the prerequisite DAG of the 19 courses used for analysis. Most of these courses are math, physics, and core MAE major courses. Generally speaking, math courses are prerequisites of physics courses which are further prerequisites of MAE major courses. We consider two cases of course availability constraints. In the first case, most courses are offered twice every academic year and in the second case, most courses are offered once every academic year. This allows us to investigate the impact of course availability on course sequence recommendation. We focus on how to recommend course sequences to complete these 19 courses as soon as possible.



First, we investigate the possible course completion states that can emerge. Although there are totally $2^{19} = 524288$ possible states, the course prerequisites and availability significantly limit the number of possible states that can emerge. Figure \ref{States} illustrates the number of possible states (i.e. the size of the state set $\mathcal{L}(t)$ defined in Section III-A). Interestingly, the course prerequisite constraint limits $\lim\limits_{t\to\infty} \mathcal{L}(t)$ as the quarters go by. In our experiment, there are totally 4880 possible completion states depending on whether the student has taken, passed or failed the course, which are significantly fewer than the possible states without any prerequisite constraints. Moreover, the maximal number of possible states are reached within a finite number of quarters. On the other hand, the course availability constraint affects how fast the set $\mathcal{L}(t)$ expands and reaches its maximal size. When courses are offered less frequently, $\mathcal{L}(t)$ expands more slowly but will eventually reach the maximal size.

\begin{figure}
  \centering
  \includegraphics[scale = 0.6]{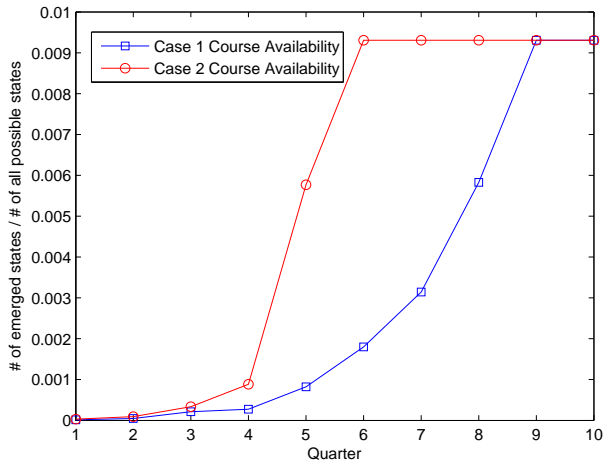}\\
  \caption{Number of Course Completion States.}\label{States}
\end{figure}

\subsection{Course sequences}
In this subsection, we apply the forward-search backward-induction algorithm presented in Sec. III to compute candidate course sequence recommendation policies to students to complete the 19 courses presented in Sec. V-C. Although the course difficulty varies across courses, for illustrative purposes, we set the course failure ratio to be the same $\epsilon = 0.1$ for all courses. Indeed, common practice tells us that teachers often curve the students grades so that the passing ratios of courses do not vary significantly across courses. Table \ref{Best1} and \ref{Best2} show the best sequences that can be obtained for the two course availability cases when the student does not fail any course.
It is worth noting that while some courses are taken after some others in the first case, the order can be reversed in the second case due to different course availability constraints. For instance, course CS 31 is taken in the first quarter before courses MATH 31B and MSE 104 in the first cases but it is taken in the third quarter after courses MATH 31B and MSE 104 in the second case. The generated course sequence is also useful for the students to decide when is good time to take extracurricular courses (e.g. the art courses). For instance, a student should focus on the curriculum courses in quarter 1, 3 and 5 since three courses need to be taken in each of these quarters while quarters 4 and 8 are good time for the student to take extracurricular courses that fall in the student's interest.

\begin{table}
\center
\caption{Best Recommended Course Sequence (Case 1 Course Availability)}\label{Best1}
\begin{tabular}{|c|c|}
  \hline
  Quarter & Recommended Courses \\
  \hline
  1 & MATH 31A, CS 31, MAE 94, CHE 20 \\
  2 & MATH 31B, MSE 104 \\
  3 & PHY 1A, MATH 32A, MATH 33B \\
  4 & PHY 1B, PHY 4AL, MATH 32B \\
  5 & PHY 1C, PHY 4BL, MATH 33A, MAE 101\\
  6 & MAE 105A, MAE 102, MAE 103 \\
  \hline
\end{tabular}
\end{table}

\begin{table}
\center
\caption{Best Recommended Course Sequence (Case 2 Course Availability)}\label{Best2}
\begin{tabular}{|c|c|}
  \hline
  Quarter & Recommended Courses \\
  \hline
  1 & MATH 31A, MAE 94, CHE 20 \\
  2 & MATH 31B, MSE 104 \\
  3 & CS 31, MATH 32A, PHY 1B \\
  4 & MATH 32B \\
  5 & PHY 1A, MATH 33A, MAE 105A\\
  6 & PHY 1B, PHY 4AL \\
  7 & PHY 1C, PHY 4BL \\
  8 & MAE 101\\
  9 & MAE 102, MAE 103\\
  \hline
\end{tabular}
\end{table}

Since students may fail the course, the best course sequence can no longer be followed by a student. Once this happens, it becomes important to determine the subsequent courses to recommend in any possible course completion state. The course sequence recommendation policy generated by our algorithm can provide us with these answers for any course failure ratio. Table \ref{Time} shows the results by running our algorithm aimed at maximizing the on-time graduation probability or minimizing the graduation time for the two course availability cases. When courses are offered more frequently, the expected graduation time is significantly reduced and the on-time graduation probability is significantly improved. The proposed framework can also be used to evaluate different allocations of teaching resources by calculating their expected on-time graduation probability and expected graduation time.

\begin{table}
\center
\caption{Completion Time Results}\label{Time}
\begin{tabular}{|c|c|c|c|}
  \hline
   & \begin{tabular}{@{}c@{}}Probability of \\Completing 19 Courses \\in 10 Quarters\end{tabular}& \begin{tabular}{@{}c@{}}Expected\\Completion\\Time \end{tabular}& \begin{tabular}{@{}c@{}}Best\\Sequence\\Time\end{tabular}\\
   \hline
  Case 1 & 39.70\% & 9.3 quarters & 9 quarters \\
  \hline
  Case 2 & 98.10\% & 6.4 quarters & 6 quarters \\
  \hline
\end{tabular}
\end{table}

\subsection{Personalized policy selection}
As mentioned in Sec. IV, depending on the context information of the students, different course sequence recommendation policies may result in different learning experience. The first step of our framework constructs course sequence recommendation policies to minimize time-to-graduation. The second step then personalizes the recommendation according to the students' context to achieve a high GPA.

Due to the limited number of student records that we have, we focus on the subsequence recommendation for a subset of 3 MAE major courses (i.e. MAE 101, MAE 103, MAE 105A). From our dataset, we observe that there are six  sequences that student use to take these 3 courses. We will use these six typical sequences to validate the proposed personalized policy selection algorithm. Note that we are not using the policies generated in the first step to make course sequence recommendations because the results on the students in the dataset cannot be known if they are not aligned with the actual sequences that the students take. Thus, instead of evaluating the joint efficacy of the policy construction and policy selection, we only evaluate the efficacy of the policy selection algorithm in this subsection.


\begin{table}
\centering
\caption{GPA statistics for the six course sequences}\label{stat}
\begin{tabular}{|c|c|c|c|c|c|c|c|}
  \hline
  \multicolumn{2}{|c|}{Sequence} & 1 & 2 & 3 & 4 & 5 & 6 \\
  \hline
  \multicolumn{2}{|c|}{Number of students} & 89 & 75 & 72 & 51 & 31 & 18 \\
  \hline
  \multirow{ 2}{*}{SAT$\leq$700}  & mean & \textbf{3.36} & 3.02 & 3.08 & 3.05 & \textbf{3.31} & 3.09 \\
  & count & 28 & 20 & 28 & 13 & 21 & 4 \\
  \hline
  \multirow{ 2}{*}{700$<$SAT$\leq$760}  & mean & 3.28 & \textbf{3.37} & 3.29 & 3.17 & 3.26 & \textbf{3.39} \\
  & count & 38 & 22 & 22 & 16 & 5 & 5 \\
  \hline
  \multirow{ 2}{*}{760$<$SAT$\leq$780}  & mean & \textbf{3.61} & 3.13 & 3.22 & 3.25 & NA & \textbf{3.50} \\
  & count & 14 & 19 & 14 & 10 & 0 & 8 \\
  \hline
  \multirow{ 2}{*}{SAT$>$780}  & mean & 3.39 & \textbf{3.45} & 3.16 & 3.33 & 3.04 & \textbf{3.9} \\
  & count & 9 & 14 & 8 & 12 & 5 & 1 \\
  \hline
\end{tabular}
\end{table}

Table \ref{stat} shows the statistics regarding these 6 subsequences. As we can see, depending on the context information (i.e. math SAT score) of the students, different course sequences yield different GPA performance. The average GPA of these students is 3.26. The evaluation of the proposed personalized policy selection will use the statistics given in Table \ref{stat}.  We compare the proposed algorithm with several benchmarks:
\begin{itemize}
  \item \textbf{Oracle}: the oracle algorithm knows the GPA statistics a priori and hence always recommends the best course sequence to each student.
  \item \textbf{Learning without Personalization}: this algorithm does not know the GPA statistics. However, when recommending course sequences, it ignores the context information of the students.
  \item \textbf{Random}: this algorithm simply recommends course sequence to students randomly.
\end{itemize}

Figure \ref{MAB} shows the average GPA that the students can obtain as the proposed algorithm recommends course sequences to more students.  Initially, the achievable GPA is low since the algorithm does not have sufficient training samples. As more training samples are provided, the performance of the algorithm improves, causing an increase in the simulated GPA for the selected course sequence. Moreover, the algorithm adaptively clusters the students according to their context information and significantly outperforms sequence recommendation that ignores the personalized context information. It is noteworthy that the Random scheme achieves a similar GPA as the actual average GPA in the dataset (i.e. 3.26). This suggests that the current practice of course sequence selection does not recognize the difference in individual students and hence, there is much room to improve the students' learning outcomes by personalizing the course sequences recommended to students.

\begin{figure}
  \centering
  \includegraphics[scale = 0.6]{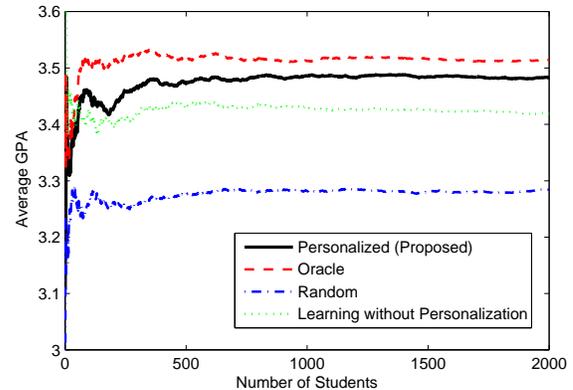}\\
  \caption{GPA performance v.s. the number of students}\label{MAB}
\end{figure}

\section{Conclusion}
In this paper, we studied the problem of personalized course sequence recommendation. The problem is solved in two steps. In the first step, we determine candidate course sequence recommendation policies that result in short time-to-graduation using a Forward-Search Backward-Induction algorithm. In the second step, we develop an online regret minimization learning algorithm to select personalized course sequence recommendation policies among the candidate policies for students aimed at maximizing students' GPA performance. Our analysis and simulation results show that the proposed personalized course sequence recommendation method is able to shorten the students' graduation time and improve students' GPAs. Our framework also has important implications on how the curriculum planner should design the curriculum and allocate teaching resources.

\section*{Acknowledgment}
The authors are indebted to Mr. John S. Toledo, Research Associate, Evaluation and Educational Assessment, UCLA Office of Instructional Development for providing us the data based on which our methods were evaluated.



%

%
%
%
%
%

\bibliographystyle{IEEEtran}
\bibliography{refs}

\end{document}